\documentclass[11pt]{article}

\usepackage[margin=1in]{geometry}
\usepackage{amsmath,amssymb,amsthm,mathtools}
\usepackage{algorithm}
\usepackage{algpseudocode}
\usepackage{microtype}
\usepackage{xcolor}
\usepackage{hyperref}
\usepackage[nameinlink,capitalise,noabbrev]{cleveref}

\definecolor{newblue}{rgb}{0.19,0.55,0.91}
\hypersetup{
    colorlinks=true,
    linkcolor=blue,
    citecolor=newblue,
    urlcolor=blue!60!black,
    pdftitle={Pairwise-Independent Dithering for Single-Stage Hadamard Quantization},
    pdfauthor={Honghao Lin},
    pdfsubject={Randomized Hadamard quantization and inner-product estimation}
}

\numberwithin{equation}{section}

\newtheorem{theorem}{Theorem}[section]
\newtheorem{lemma}[theorem]{Lemma}
\newtheorem{proposition}[theorem]{Proposition}
\newtheorem{corollary}[theorem]{Corollary}
\theoremstyle{definition}

\crefname{theorem}{Theorem}{Theorems}
\Crefname{theorem}{Theorem}{Theorems}
\crefname{lemma}{Lemma}{Lemmas}
\Crefname{lemma}{Lemma}{Lemmas}
\crefname{proposition}{Proposition}{Propositions}
\Crefname{proposition}{Proposition}{Propositions}
\crefname{corollary}{Corollary}{Corollaries}
\Crefname{corollary}{Corollary}{Corollaries}
\crefname{definition}{Definition}{Definitions}
\Crefname{definition}{Definition}{Definitions}
\crefname{remark}{Remark}{Remarks}
\Crefname{remark}{Remark}{Remarks}

\DeclarePairedDelimiter\norm{\lVert}{\rVert}
\DeclarePairedDelimiter\abs{\lvert}{\rvert}
\newcommand{\R}{\mathbb{R}}
\newcommand{\E}{\mathbb{E}}
\newcommand{\Sph}{\mathbb{S}}
\newcommand{\eps}{\varepsilon}
\newcommand{\dd}{\,\mathrm{d}}
\newcommand{\Unif}{\operatorname{Unif}}
\newcommand{\diag}{\operatorname{diag}}
\newcommand{\Q}{\mathsf{Q}}
\newcommand{\enc}{\mathsf{Enc}}
\newcommand{\dec}{\mathsf{Dec}}

\title{Pairwise-Independent Dithering for
Single-Stage Hadamard Quantization}
\date{}

\author{
Honghao Lin\footnote{Google Research, Carnegie Mellon University / Texas A\&M University. \texttt{honghaol3010@gmail.com}}
\and
Vahab Mirrokni\footnote{Google Research. \texttt{mirrokni@google.com}}
\and
David P. Woodruff\footnote{Google Research and Carnegie Mellon University. \texttt{dpwoodru@gmail.com }}
}

\begin{document}
\maketitle

\begin{abstract}
Quantizing high-dimensional vectors is fundamental to similarity search,
distributed learning, and model compression.  Feng, Indyk, Kapralov, Krachun,
and Prokhorov established sharp guarantees for an unbiased dithered quantizer
based on a randomized Hadamard transform~\cite{feng2026provable}.  Their
$1/d$-scale inner-product estimator, however, uses a second randomized
transform and residual quantization, increasing both communication and the
leading constant in the proved bound.  We show that this extra stage is
unnecessary: pairwise-independent dithers across Hadamard coordinates suffice.
The resulting unbiased single-stage estimator uses $b$ bits per coordinate and
achieves
\[
    \E\abs{\langle y,\widehat{x}-x\rangle}^{2}
    \leq
    \left(\frac{3\pi\sqrt{3}}{2}+o(1)\right)
    \frac{\norm{y}_2^2}{d\,4^b},
\]
as $b\to\infty$, with a dimension-free $o(1)$ term uniform over unit inputs
and fixed queries.
Compared with the two-stage construction of Feng et al., it eliminates the
residual-stage $O(d)$-bit payload and reduces the leading upper-bound constant
by a factor of approximately $5.93$.

The proof was first obtained using a fully automated Gemini-based agentic
system developed internally at Google.  The authors have verified the proof and
edited it for clarity of presentation.
\end{abstract}

\section{Introduction}
\label{sec:introduction}

Vector quantization represents a high-dimensional vector by a short binary
string while retaining the geometric information needed by downstream tasks.
At $b$ bits per coordinate, a vector quantizer consists of an encoder
$\enc:\R^d\to\{0,1\}^{bd}$ and a decoder
$\dec:\{0,1\}^{bd}\to\R^d$; on input $x\in\R^d$, it returns the reconstruction
$\widehat{x}=\dec(\enc(x))$.  The encoder and decoder may use shared
randomness, which is not counted toward the representation length.  The goal is
to preserve geometric information about $x$, such as reconstruction accuracy,
distances, and inner products, while keeping the representation small.  This
primitive is used in similarity search
\cite{charikar2002,jegou2011,datar2004,dong2008,gao2024,gao2025optimal},
federated learning
\cite{drive2021,vargaftik2022,benbasat2024quick}, efficient transformers and
KV-cache compression \cite{zandieh2026}, and model compression
\cite{gholami2022,ashkboos2024,tseng2024}.

A widely used data-oblivious template has the form
$\enc(x)=\operatorname{quant}(Rx)$, where $R$ is a random rotation matrix and
$\operatorname{quant}(\cdot)$ is often a scalar quantizer applied
coordinatewise \cite{charikar2002,datar2004,dong2008,ashkboos2024}.  Broader
variants replace the scalar rule by structured vector codebooks
\cite{andoni2015,tseng2024}.  The rotation distribution and the codebook are
chosen independently of the dataset, so no data-dependent training is required.
Moreover, the method is amenable to analysis because each coordinate of $Rx$ is
approximately normally distributed when $R$ is a uniformly random rotation
\cite{feng2026provable}.
Variants of this template include DRIVE and EDEN for distributed mean estimation
\cite{drive2021,vargaftik2022}, RaBitQ for distance and inner-product
estimation \cite{gao2024,gao2025optimal}, and TurboQuant
\cite{zandieh2026}.

Despite these advantages, applying a dense random rotation matrix $R$ requires
$\Theta(d^2)$ time, which can be prohibitively expensive in high dimensions.  A
standard
way to reduce this cost is to use randomized Hadamard transforms
\cite{ailon2009,sarlos2006,woolfe2008,krahmer2011,halko2011}.  Specifically,
one replaces $R$ with $HD$, where $H$ is a normalized Hadamard matrix and $D$
is a diagonal matrix of independent random signs.  For the Walsh--Hadamard
choice, the resulting transform can be applied in $O(d\log d)$ time using the
fast Walsh--Hadamard algorithm.
However, this speedup comes at an analytical cost: $HD$ is discrete, and each
coordinate of $HDx$ is an input-dependent weighted sum of the random signs;
for worst-case inputs, this sum need not be approximately normal
\cite{feng2026provable}.

Feng, Indyk, Kapralov, Krachun, and Prokhorov
\cite{feng2026provable} overcame this difficulty for vector distortion by
combining a single randomized Hadamard transform with a dithered,
Gaussian-companded scalar quantizer.  At $b$ bits per coordinate, their
reconstruction is exactly unbiased and satisfies, uniformly over unit inputs,
\[
    \E\norm{\widehat{x}-x}_2^2
    \leq
    \left(\frac{\pi\sqrt{3}}{2}+o(1)\right)4^{-b},
\]
as $b\to\infty$.  For fixed-query inner products, their separate construction
uses a closely related two-stage scheme: after an initial randomized Hadamard
quantization, it projects the reconstruction onto the unit ball and quantizes
the residual after a fresh randomized Hadamard transform.  The residual stage
yields the desired $1/d$ scaling, but it also adds communication and enlarges
the leading constant in the proved bound.

In this work, we show that the residual stage is unnecessary.  In the base
quantizer, one scalar dither is shared across all Hadamard coordinates, so the
coordinate errors can remain correlated even though each is centered.  We
replace the shared dither by pairwise-independent dithers.  Conditioned on
$D$, the errors
are then centered and pairwise independent, so every off-diagonal term in the
inner-product second moment vanishes.  Thus pairwise independence suffices for
the desired second-moment bound.

\subsection{Related work}

\paragraph{Random rotations and coordinate-wise quantization.}
Random projection followed by scalar quantization is classical in similarity
estimation and locality-sensitive hashing
\cite{charikar2002,datar2004,dong2008}.  In distributed mean estimation,
Suresh et al. combined a randomized Hadamard transform with stochastic
quantization \cite{suresh2017}, while DRIVE and EDEN developed rotation-based
schemes that include unbiased reconstruction and flexible rates
\cite{drive2021,vargaftik2022}.  Thus this template, including unbiased
variants, predates TurboQuant.  In vector search, RaBitQ and its multi-bit
extension give rigorous distance and inner-product guarantees, including an
asymptotically optimal space--error tradeoff
\cite{gao2024,gao2025optimal}, consistent with the general compression bounds
of Alon and Klartag \cite{alon2017optimal}.  TurboQuant later combined a
dense-rotation MSE quantizer with residual quantization
\cite{zandieh2026}.  Its relationship to DRIVE/EDEN and RaBitQ is discussed in
\cite{benbasat2026note,gao2026comparison}.

\paragraph{Dithering.}
Randomized offsets have a long history in quantization theory, where they are
used to control bias and the dependence of quantization noise on the input
\cite{schuchman1964,zamir1992}.  Dithered scalar grids have also been combined
with random projections for similarity and distance estimation
\cite{datar2004,dong2008}.  Feng et al.~\cite{feng2026provable} use one scalar
offset shared across all transformed coordinates.  We leave their scalar
quantizer unchanged and modify only the dependence structure of the offsets.

\paragraph{Fast structured rotations.}
Randomized Hadamard transforms are central to fast embedding and randomized
linear-algebra methods \cite{ailon2009,halko2011,krahmer2011,tropp2011}.
Approximation of dense random rotations by one or more such transforms has
been studied for a single transform \cite{cherapanamjeri2022}, two structured
blocks \cite{zilca2026}, and multiple transforms with blockwise dependence
control \cite{benbasat2026}.  Direct single-transform quantization analyses
appear in distributed mean estimation \cite{benbasat2024quick} and, using
dithering, in the sharp high-rate result of Feng et al.
\cite{feng2026provable}.

\paragraph{Relation to Feng et al.}
Our scalar quantizer is exactly the unbiased Gaussian-companded construction
of Feng et al.~\cite{feng2026provable}.  We invoke their reconstruction map,
averaging identity, and scalar error estimates directly.  Our new
ingredients are the pairwise-independent dither family, conditional covariance
diagonalization, and the weighted mixed-Rademacher estimate for the remaining
diagonal terms.

\subsection{Our result}

We work first in a dimension $d$ admitting a normalized Hadamard matrix; in
particular, one may take $d$ to be a power of two.
Let $B=2^b$.
The encoder and decoder share the signs in $D$ and two independent random
variables $V_0,V_1\sim\Unif[0,1)$.
For coordinate $i\in\{1,\dots,d\}$, define
\begin{equation}
    U_i=(V_0+iV_1)\bmod 1.
    \label{eq:intro-dither}
\end{equation}
The family $(U_i)_{i=1}^d$ is pairwise independent and each $U_i$ is uniform
on $[0,1)$.
The quantizer applies the scalar unbiased reconstruction rule with dither
$U_i$ to the $i$th coordinate of $\sqrt{d}HDx$ and transmits one symbol from
$\{0,\dots,B-1\}$ per coordinate.

\begin{theorem}[Main theorem]
\label[theorem]{thm:main}
Let $d$ admit a normalized Hadamard matrix, let $b\geq2$ be an integer, and
set $B=2^b$.
For every fixed $x\in\Sph^{d-1}$ and $y\in\R^d$, the single-stage quantizer in
\cref{alg:quantizer}, with $D$ independent of $(V_0,V_1)$, satisfies
\begin{equation}
    \E[\widehat{x}]=x,
    \label{eq:main-unbiased}
\end{equation}
and
\begin{equation}
    \E\abs{\langle y,\widehat{x}-x\rangle}^{2}
    \leq
    \left(\frac{3\pi\sqrt{3}}{2}+o_B(1)\right)
    \frac{\norm{y}_2^2}{dB^2}.
    \label{eq:main-mse}
\end{equation}
The term $o_B(1)$ tends to zero as $B\to\infty$ uniformly over all dimensions
$d$ admitting a normalized Hadamard matrix, all $x\in\Sph^{d-1}$, and all
nonzero $y\in\R^d$.
The communicated representation consists of exactly $db$ bits, excluding the
shared randomness.
\end{theorem}

The two-stage inner-product construction of Feng et
al.~\cite[Theorem~1.2]{feng2026provable} has proved leading constant
\[
    13\left(\frac{\pi\sqrt{3}}{2}+1\right)\approx48.369
\]
and can be encoded using at most
\[
    db+\left(3+\frac{1}{2\ln2}\right)d
    +O(\log(b+\log d))
\]
bits.  Thus our single-stage estimator removes the residual payload.  Relative
to this stated upper bound, its communication guarantee is smaller by
approximately $3.72d$ bits, up to lower-order terms.  It also reduces the proved
leading constant from approximately $48.369$ to $8.162$, an improvement by a
factor of approximately $5.93$, while preserving the vector mean-squared
distortion guarantee of the base quantizer.

\subsection{Technical overview}

We use the scalar quantizer introduced by Feng et
al.~\cite{feng2026provable}.  It first maps a scalar $t$ to quantile space via
$F(t)=\Phi(t/\sqrt{3})$ and then quantizes $F(t)$ on a uniformly shifted
$B$-level grid.  The decoder maps the selected grid point back to the real
line using a modified inverse-compander $G_B$.  If $\Q_U(t)$ denotes the
reconstruction with dither $U$, the construction satisfies the exact averaging
identity $\E_U[\Q_U(t)]=t$ and, in the central region, the variance bound
\[
    \E_U\bigl[(\Q_U(t)-t)^2\bigr]
    \leq
    \left(\frac{\pi}{2}+o_B(1)\right)B^{-2}e^{t^2/3}.
\]
A coarser global bound handles the tails.  Feng et al. use a single scalar
dither across all Hadamard coordinates.  In their separate inner-product
construction, the first-stage reconstruction is projected onto the unit ball
and the resulting residual is quantized after a fresh randomized Hadamard
transform; the fresh transform provides the decorrelation needed for the $1/d$
variance reduction.  Our construction instead builds this decorrelation into
the unbiased single-stage quantizer: we replace the shared dither by the
pairwise-independent family $(U_i)_{i=1}^d$ and use the reconstruction directly.

The role of pairwise independence becomes clear upon expanding the
inner-product error.  Set
$z=\sqrt{d}HDx$, $e_i=\Q_{U_i}(z_i)-z_i$, and, for $y\neq0$,
$w_i=\sqrt{d}(HDy)_i/\norm{y}_2$.  Orthogonality gives
\[
    \langle y,\widehat{x}-x\rangle
    =\frac{\norm{y}_2}{d}\sum_{i=1}^d w_i e_i.
\]
Once $D$ is fixed, the quantities $z_i,w_i$ are deterministic and $e_i$
depends only on $U_i$.  The averaging identity centers each error, while
pairwise independence eliminates the cross terms:
$\E[e_i\mid D]=0$ and $\E[e_ie_j\mid D]=0$ for $i\neq j$.  Since the target is
a second moment, no higher-order independence is needed.  Averaging over $D$
therefore gives
\[
    \E\abs{\langle y,\widehat{x}-x\rangle}^2
    =\frac{\norm{y}_2^2}{d^2}
      \sum_{i=1}^d \E[w_i^2e_i^2].
\]
There are $d$ summands and a prefactor $d^{-2}$, so a dimension-free
$O(B^{-2})$ bound for each summand yields the desired factor $1/d$.

The remaining issue is therefore a weighted scalar estimate.  For a fixed
coordinate, $(z_i,w_i)$ has the same joint law as
$X=\sum_j a_j\eps_j$ and $Y=\sum_j c_j\eps_j$, where the two sums use the same
Rademacher signs, $\norm{a}_2=\norm{c}_2=1$,
$a_j=\sqrt d H_{ij}x_j$, and
$c_j=\sqrt d H_{ij}y_j/\norm{y}_2$.  The weight $Y^2$ carries the query
direction and prevents us from applying the unweighted scalar bound directly.
On the central event, the variance estimate above produces the mixed moment
$\E[Y^2e^{X^2/3}]$.  On the complement, the global bound, Markov's inequality
applied to the tilted exponential-moment estimate, and layer-cake integration
give a uniform $o_B(1)B^{-2}$ contribution.
Together these estimates yield
\[
    \E\left[Y^2\bigl(\Q_U(X)-X\bigr)^2\right]
    \leq
    \left(
        \frac{\pi}{2}\E[Y^2e^{X^2/3}]+o_B(1)
    \right)B^{-2}.
\]
The sums $X$ and $Y$ are generally dependent because they use the same signs.
Nevertheless, a direct exponential-moment calculation for the underlying
Rademacher signs, followed by a Gaussian-mixture representation of
$e^{X^2/3}$, gives
\[
    \E[Y^2e^{\lambda X}]
    \leq(1+\lambda^2)e^{\lambda^2/2}
    \quad\Longrightarrow\quad
    \E[Y^2e^{X^2/3}]\leq3\sqrt{3}.
\]
Each diagonal term is therefore at most
$\left(3\pi\sqrt{3}/2+o_B(1)\right)B^{-2}$, and summing over the $d$
coordinates proves \cref{eq:main-mse}.

\section{Preliminaries}
\label{sec:preliminaries}

Throughout, $x\in\Sph^{d-1}$ is a fixed input and $y\in\R^d$ is a fixed
query.  For an integer bit width $b\geq2$, set $B=2^b$ and
$\delta=(B-1)^{-1}$.  Unless stated otherwise, expectations are over all
randomness in the construction.

\subsection{Randomized Hadamard transforms}

A normalized Hadamard matrix is a matrix $H\in\R^{d\times d}$ satisfying
$H_{ij}\in\{\pm d^{-1/2}\}$ and $H^\top H=I_d$.  The normalized
Walsh--Hadamard matrix provides such an $H$ whenever $d$ is a power of two.
Let
\[
    D=\diag(\eps_1,\dots,\eps_d),
\]
where $\eps_1,\dots,\eps_d$ are independent random signs.  We refer to $HD$ as
a randomized Hadamard transform.

An arbitrary ambient dimension $n$ can be handled by zero-padding the input to
$d=2^{\lceil\log_2 n\rceil}<2n$, applying the scheme in $\R^d$, and discarding
the final $d-n$ coordinates after decoding; the query is zero-padded for the
error analysis.  The resulting representation uses $db<2nb$ bits, and the
directional error bound holds with the padded dimension $d$.

\subsection{Gaussian-companded scalar quantization}

Let $\Phi$ denote the standard Gaussian distribution function.  Following
Feng et al.~\cite{feng2026provable}, define the compander
\[
    F(t)=\Phi\left(\frac{t}{\sqrt{3}}\right),
    \qquad
    F'(t)=\frac{1}{\sqrt{6\pi}}e^{-t^2/6}.
\]
The identity
\begin{equation}
    \frac{1}{12F'(t)^2}=\frac{\pi}{2}e^{t^2/3}
    \label{eq:F-identity}
\end{equation}
will determine the leading scalar error constant.

Let $G_B$ be the reconstruction map constructed by Feng et
al.~\cite[Lemma~A.8]{feng2026provable}:
\[
    G_B:\left[-\frac{\delta}{2},1+\frac{\delta}{2}\right]\to\R.
\]
It satisfies the averaging identity
\[
    \frac{1}{\delta}
    \int_{r-\delta/2}^{r+\delta/2}G_B(s)\dd s
    =F^{-1}(r),
    \qquad r\in(0,1).
\]
For $U\sim\Unif[0,1)$ and $t\in\R$, define the transmitted index and scalar
reconstruction by
\[
    J_U(t)=\lfloor(B-1)F(t)-U\rfloor+1,
    \qquad
    \Q_U(t)=G_B\left(\frac{J_U(t)+U-1/2}{B-1}\right).
\]
Since $F(t)\in(0,1)$, the index $J_U(t)$ lies in
$\{0,\dots,B-1\}$.  For fixed $t$, the argument of $G_B$ is distributed as
$F(t)+V\delta$, where $V\sim\Unif[-1/2,1/2]$.  The averaging identity therefore
gives
\begin{equation}
    \E_U[\Q_U(t)]=t
    \qquad(t\in\R).
    \label{eq:scalar-unbiased}
\end{equation}

For this reconstruction rule, the following estimates are direct consequences
of Lemma~A.8 and the proof of Lemma~A.9 in \cite{feng2026provable}.

\begin{proposition}[Scalar error bounds]
\label[proposition]{prop:scalar}
Let
\[
    \psi_B(t)=\E_U\bigl[(\Q_U(t)-t)^2\bigr],
    \qquad
    M_B=\sqrt{5\log B}.
\]
Then the following bounds hold.
\begin{enumerate}
    \item Uniformly for $\abs{t}\leq M_B$,
    \begin{equation}
        \psi_B(t)
        \leq
        (1+o_B(1))\frac{\delta^2}{12F'(t)^2}
        =(1+o_B(1))\frac{\pi\delta^2}{2}e^{t^2/3}.
        \label{eq:scalar-central}
    \end{equation}

    \item There is an absolute constant $C$ such that, for every $t\in\R$,
    \begin{equation}
        \psi_B(t)
        \leq
        C\left(1+t^2+\frac{\delta^2}{F'(t)^2}\right)
        \leq C\left(1+t^2+\delta^2e^{t^2/3}\right).
        \label{eq:scalar-global}
    \end{equation}
\end{enumerate}
\end{proposition}

\section{The quantizer and its error decomposition}
\label{sec:quantizer}

We now define the single-stage quantizer and derive the conditional moment
identity that reduces its analysis to a weighted estimate for Rademacher sums.

\subsection{Pairwise-independent dithers}

We retain the scalar quantization rule from \cref{sec:preliminaries}, but replace
the single dither shared across Hadamard coordinates by a pairwise-independent
family.  Two independent uniform random variables suffice to generate it.

\begin{lemma}[Affine dither construction]
\label[lemma]{lem:pairwise-dither}
Let $V_0,V_1$ be independent random variables, each uniform on $[0,1)$ and
independent of $D$.  For $i=1,\dots,d$, define
\begin{equation}
    U_i=(V_0+iV_1)\bmod 1,
    \label{eq:affine-dither}
\end{equation}
where $a\bmod1$ denotes the fractional part of $a$.
Each $U_i$ is uniform on $[0,1)$, and $U_i$ and $U_j$ are independent whenever
$i\neq j$.
\end{lemma}

\begin{proof}
Conditional on $V_1=v$, the variable $U_i$ is a modulo-one translate of $V_0$
and is therefore uniform on $[0,1)$.  Since this conditional distribution does
not depend on $v$, the variables $U_i$ and $V_1$ are independent.

Now fix $i\neq j$ and set $W=((j-i)V_1)\bmod1$.  The variable $W$ is uniform
on $[0,1)$: after partitioning $[0,1)$ into $\abs{j-i}$ equal intervals, the
map $v\mapsto((j-i)v)\bmod1$ covers $[0,1)$ once on each interval, with the
orientation reversed when $j<i$.  Since $U_i$ is independent of $V_1$, it is
also independent of $W$.  Moreover,
\[
    U_j=(U_i+W)\bmod1.
\]
Conditioning on $U_i$ leaves $W$ uniform, so the right-hand side is conditionally
uniform on $[0,1)$.  Its distribution does not depend on $U_i$, proving that
$U_i$ and $U_j$ are independent.
\end{proof}

The dither sequence can be generated in $O(d)$ time: after computing
$U_1=(V_0+V_1)\bmod1$, use $U_{i+1}=(U_i+V_1)\bmod1$.

\subsection{Single-stage encoder and decoder}

Applying the scalar rule coordinatewise with the dithers above gives
\cref{alg:quantizer}.  The scheme coincides with the base vector quantizer of
Feng et al.~\cite{feng2026provable}, except that its shared scalar dither is
replaced by the affine family in \cref{lem:pairwise-dither}.  All randomness
sampled in the algorithm is shared by the encoder and decoder.

\begin{algorithm}[H]
\caption{Pairwise-independent single-stage Hadamard quantizer}
\label{alg:quantizer}
\begin{algorithmic}[1]
\Require Dimension $d$, integer bit width $b\geq2$, normalized Hadamard matrix
$H\in\R^{d\times d}$
\State Set $B\gets2^b$.
\State Sample independent Rademacher signs $\eps_1,\dots,\eps_d$ and set
$D\gets\diag(\eps_1,\dots,\eps_d)$.
\State Sample independent $V_0,V_1\sim\Unif[0,1)$, independently of $D$.
\For{$i=1,\dots,d$}
    \State $U_i\gets(V_0+iV_1)\bmod1$.
\EndFor
\Statex
\Procedure{$\enc$}{$x\in\Sph^{d-1}$}
    \State $z\gets\sqrt{d}\,HDx$.
    \For{$i=1,\dots,d$}
        \State $j_i\gets\lfloor(B-1)F(z_i)-U_i\rfloor+1$.
    \EndFor
    \State \Return $j=(j_1,\dots,j_d)\in\{0,\dots,B-1\}^d$.
\EndProcedure
\Statex
\Procedure{$\dec$}{$j$}
    \For{$i=1,\dots,d$}
        \State $q_i\gets G_B\bigl((j_i+U_i-1/2)/(B-1)\bigr)$.
    \EndFor
    \State $\widehat{x}\gets DH^\top(q/\sqrt{d})$.
    \State \Return $\widehat{x}$.
\EndProcedure
\end{algorithmic}
\end{algorithm}

Since each $j_i$ belongs to an alphabet of size $B=2^b$, fixed-length encoding
uses exactly $b$ bits per coordinate.  The transmitted message therefore has
exactly $db$ bits; the shared randomness is not counted.

\subsection{Conditional first and second moments}
\label{sec:diagonalization}

Define the transformed input and scalar errors by
\begin{equation}
    z_i=\sqrt{d}(HDx)_i,
    \qquad
    e_i=\Q_{U_i}(z_i)-z_i.
    \label{eq:z-e-definition}
\end{equation}
For $y\neq0$, define the normalized transformed query coefficients
\begin{equation}
    w_i=\frac{\sqrt{d}(HDy)_i}{\norm{y}_2}.
    \label{eq:w-definition}
\end{equation}

\begin{proposition}[Conditional unbiasedness]
\label[proposition]{prop:conditional-unbiasedness}
For every fixed $x\in\Sph^{d-1}$,
\[
    \E[\widehat{x}\mid D]=x.
\]
\end{proposition}

\begin{proof}
Coordinatewise scalar unbiasedness gives
$\E[q\mid D]=z=\sqrt{d}HDx$.  Hence
\[
    \E[\widehat{x}\mid D]
    =DH^\top\frac{\E[q\mid D]}{\sqrt{d}}
    =DH^\top HDx=x.
\]
\end{proof}

\begin{proposition}[Conditional second moment]
\label[proposition]{prop:conditional-second-moment}
Conditioned on $D$, the errors $e_1,\dots,e_d$ are centered and pairwise
independent.  Consequently, for every fixed $y\neq0$,
\begin{equation}
    \E\left[
        \abs{\langle y,\widehat{x}-x\rangle}^{2}
        \,\middle|\,D
    \right]
    =\frac{\norm{y}_2^2}{d^2}
      \sum_{i=1}^d w_i^2\psi_B(z_i).
    \label{eq:conditional-diagonalization}
\end{equation}
\end{proposition}

\begin{proof}
The decoder gives
\[
    \widehat{x}-x
    =DH^\top\left(\frac{q-z}{\sqrt{d}}\right)
    =DH^\top\left(\frac{e}{\sqrt{d}}\right),
\]
and therefore
\begin{equation}
    \langle y,\widehat{x}-x\rangle
    =\frac{\norm{y}_2}{d}\sum_{i=1}^d w_i e_i.
    \label{eq:error-representation}
\end{equation}
Once $D$ is fixed, every $z_i$ and $w_i$ is deterministic.
The error $e_i$ depends only on $U_i$.
By \cref{lem:pairwise-dither}, the errors are therefore pairwise independent
under the remaining randomness.
Moreover, \cref{eq:scalar-unbiased} gives
\[
    \E[e_i\mid D]=0.
\]
For $i\neq j$, it follows that
\[
    \E[e_ie_j\mid D]
    =\E[e_i\mid D]\E[e_j\mid D]=0
    \qquad(i\neq j).
\]
Since $\E[e_i^2\mid D]=\psi_B(z_i)$, expanding the square in
\cref{eq:error-representation} proves \cref{eq:conditional-diagonalization}.
\end{proof}

We next express the diagonal terms in \cref{eq:conditional-diagonalization} as
weighted errors of Rademacher sums.  Fix $y\neq0$ and a row $i$ of $H$, and set
\[
    \eta_j=\sqrt{d}H_{ij}\eps_j,
    \qquad j=1,\dots,d.
\]
These are independent Rademacher signs, and
\begin{equation}
    z_i=\sum_{j=1}^d x_j\eta_j,
    \qquad
    w_i=\sum_{j=1}^d\frac{y_j}{\norm{y}_2}\eta_j.
    \label{eq:row-rademacher}
\end{equation}
Both coefficient vectors in \cref{eq:row-rademacher} have Euclidean norm one.
Averaging \cref{eq:conditional-diagonalization} over $D$ therefore gives
\begin{equation}
    \E\abs{\langle y,\widehat{x}-x\rangle}^{2}
    =\frac{\norm{y}_2^2}{d^2}
      \sum_{i=1}^d\E_D[w_i^2\psi_B(z_i)],
    \label{eq:rademacher-reduction}
\end{equation}
where every summand has the form $\E[Y^2\psi_B(X)]$ for two normalized
Rademacher sums built from the same signs.

\section{Scalar error for Rademacher sums}
\label{sec:rademacher}

Throughout this section, $\eta_1,\dots,\eta_d$ are independent Rademacher
random variables,
\[
    X=\sum_{j=1}^d a_j\eta_j,
    \qquad
    Y=\sum_{j=1}^d c_j\eta_j,
    \qquad
    \norm{a}_2=\norm{c}_2=1,
\]
and $U\sim\Unif[0,1)$ is independent of the signs.  The reduction in
\cref{eq:rademacher-reduction} leads to the following estimate.

\begin{lemma}[Weighted scalar estimate]
\label[lemma]{lem:weighted-scalar}
Uniformly over the dimension and all unit coefficient vectors,
\begin{equation}
    \E\left[Y^2(\Q_U(X)-X)^2\right]
    \leq
    \left(
        \frac{\pi}{2}\E[Y^2e^{X^2/3}]+o_B(1)
    \right)B^{-2}.
    \label{eq:weighted-scalar}
\end{equation}
Consequently,
\begin{equation}
    \E\left[Y^2(\Q_U(X)-X)^2\right]
    \leq
    \left(\frac{3\pi\sqrt{3}}{2}+o_B(1)\right)B^{-2}.
    \label{eq:weighted-scalar-consequence}
\end{equation}
\end{lemma}

We first establish the exponential moment bounds used in its proof.

\subsection{Exponential moment bounds}

\begin{lemma}[Exponential moment bounds]
\label[lemma]{lem:exponential-moments}
For every $\lambda\in\R$,
\begin{equation}
    \E\left[Y^2e^{\lambda X}\right]
    \leq(1+\lambda^2)e^{\lambda^2/2}.
    \label{eq:tilted-second-moment}
\end{equation}
Moreover,
\begin{equation}
    \E\left[Y^2e^{X^2/3}\right]\leq3\sqrt{3},
    \label{eq:mixed-exponential}
\end{equation}
and every normalized Rademacher sum $X$ satisfies
\begin{equation}
    \E e^{X^2/3}\leq\sqrt{3}.
    \label{eq:unweighted-exponential}
\end{equation}
\end{lemma}

\begin{proof}
The moment generating function of $X$ satisfies
\begin{equation}
    M(\lambda):=\E e^{\lambda X}
    =\prod_{j=1}^d\cosh(\lambda a_j)
    \leq e^{\lambda^2/2}.
    \label{eq:rademacher-mgf}
\end{equation}
Set $t_j=\tanh(\lambda a_j)$.  Independence of the signs and direct expansion
of $Y^2$ give
\begin{align*}
    \frac{\E[Y^2e^{\lambda X}]}{M(\lambda)}
    &={}
    \sum_j c_j^2(1-t_j^2)
      +\left(\sum_jc_jt_j\right)^2\\
    &\leq1+\sum_jt_j^2
    \leq1+\lambda^2,
\end{align*}
where we used Cauchy--Schwarz and $\abs{\tanh u}\leq\abs{u}$.
Together with \cref{eq:rademacher-mgf}, this proves
\cref{eq:tilted-second-moment}.

Let $g\sim\mathcal{N}(0,1)$ be independent of the Rademacher variables.  The Gaussian
moment generating function gives
\[
    e^{X^2/3}=\E_g e^{\sqrt{2/3}\,gX}.
\]
Tonelli's theorem and \cref{eq:tilted-second-moment}, applied conditionally on
$g$ with $\lambda=\sqrt{2/3}\,g$, yield
\begin{align*}
    \E\left[Y^2e^{X^2/3}\right]
    &={}
    \E_g\E\left[Y^2e^{\sqrt{2/3}\,gX}\right]\\
    &\leq
    \E_g\left[\left(1+\frac{2}{3}g^2\right)e^{g^2/3}\right].
\end{align*}
For a standard Gaussian,
\[
    \E e^{g^2/3}=\sqrt{3},
    \qquad
    \E[g^2e^{g^2/3}]=3\sqrt{3}.
\]
This proves \cref{eq:mixed-exponential}.  Finally,
\[
    \E e^{X^2/3}
    =\E_g\E e^{\sqrt{2/3}\,gX}
    \leq\E_g e^{g^2/3}
    =\sqrt{3}
\]
by \cref{eq:rademacher-mgf}, proving \cref{eq:unweighted-exponential}.
\end{proof}

\subsection{Proof of the weighted estimate}

\begin{proof}[Proof of \cref{lem:weighted-scalar}]
Recall
\[
    \psi_B(t)=\E_U[(\Q_U(t)-t)^2],
    \qquad
    M_B=\sqrt{5\log B}.
\]
We split $\E[Y^2\psi_B(X)]$ over the central event
$\{\abs{X}\leq M_B\}$ and its complement.

On the central event, \cref{eq:scalar-central,eq:F-identity} give
\begin{align}
    B^2\E\left[Y^2\psi_B(X)\mathbf{1}_{\{\abs{X}\leq M_B\}}\right]
    &\leq
    (1+o_B(1))\frac{\pi}{2}
    \E\left[Y^2e^{X^2/3}\mathbf{1}_{\{\abs{X}\leq M_B\}}\right]
    \notag\\
    &\leq
    (1+o_B(1))\frac{\pi}{2}\E[Y^2e^{X^2/3}],
    \label{eq:weighted-central}
\end{align}
where $B^2\delta^2=1+O(B^{-1})$.  The bound
\cref{eq:mixed-exponential} makes the resulting $o_B(1)$ uniform.

It remains to show that the tail contribution is $o(B^{-2})$ uniformly.
For $s>0$, Markov's inequality and \cref{eq:tilted-second-moment}, with
$\lambda=s$, imply
\[
    \E[Y^2\mathbf{1}_{\{X>s\}}]
    \leq e^{-s^2}\E[Y^2e^{sX}]
    \leq(1+s^2)e^{-s^2/2}.
\]
The same argument applied to $-X$ gives
\begin{equation}
    T(s):=\E[Y^2\mathbf{1}_{\{\abs{X}>s\}}]
    \leq2(1+s^2)e^{-s^2/2}.
    \label{eq:weighted-tail-function}
\end{equation}
Layer-cake integration gives, for $M>0$,
\begin{align}
    \E[Y^2X^2\mathbf{1}_{\{\abs{X}>M\}}]
    &=M^2T(M)+2\int_M^\infty sT(s)\dd s
    \notag\\
    &\leq C(1+M^4)e^{-M^2/2},
    \label{eq:weighted-poly-tail}
\end{align}
and
\begin{align}
    \E[Y^2e^{X^2/3}\mathbf{1}_{\{\abs{X}>M\}}]
    &=e^{M^2/3}T(M)
      +\int_M^\infty\frac{2s}{3}e^{s^2/3}T(s)\dd s
    \notag\\
    &\leq C(1+M^4)e^{-M^2/6}.
    \label{eq:weighted-exp-tail}
\end{align}
The constant $C$ is absolute in both displays.

Taking $M=M_B$ in \cref{eq:weighted-poly-tail}, combining it with
\cref{eq:weighted-tail-function}, and using $M_B^2=5\log B$, gives
\begin{equation}
    B^2\E\left[Y^2(1+X^2)
        \mathbf{1}_{\{\abs{X}>M_B\}}\right]=o_B(1).
    \label{eq:weighted-poly-tail-final}
\end{equation}
Similarly, \cref{eq:weighted-exp-tail} gives
\[
    \E\left[Y^2e^{X^2/3}\mathbf{1}_{\{\abs{X}>M_B\}}\right]
    =o_B(1),
\]
uniformly over the coefficient vectors and dimension.  The global scalar
estimate \cref{eq:scalar-global} therefore yields
\begin{align*}
    B^2\E\left[Y^2\psi_B(X)
        \mathbf{1}_{\{\abs{X}>M_B\}}\right]
    &\leq
    CB^2\E\left[Y^2(1+X^2)
        \mathbf{1}_{\{\abs{X}>M_B\}}\right]\\
    &\quad
    +CB^2\delta^2\E\left[Y^2e^{X^2/3}
        \mathbf{1}_{\{\abs{X}>M_B\}}\right]\\
    &=o_B(1).
\end{align*}
Combining this with \cref{eq:weighted-central} proves
\cref{eq:weighted-scalar}.  The consequence
\cref{eq:weighted-scalar-consequence} follows from
\cref{eq:mixed-exponential}.
\end{proof}

\begin{corollary}[Unweighted scalar estimate]
\label[corollary]{cor:unweighted-scalar}
Let $\eta_1,\dots,\eta_d$ be independent Rademacher variables, let
$X=\sum_j a_j\eta_j$ with $\norm{a}_2=1$, and let
$U\sim\Unif[0,1)$ be independent of the signs.  Then, uniformly over the
dimension and all unit coefficient vectors,
\begin{equation}
    \E[(\Q_U(X)-X)^2]
    \leq
    \left(\frac{\pi\sqrt{3}}{2}+o_B(1)\right)B^{-2}.
    \label{eq:unweighted-scalar}
\end{equation}
\end{corollary}

\begin{proof}
Let $\eta_0$ be a Rademacher sign independent of $U,\eta_1,\dots,\eta_d$, and write
$\widetilde X=X+0\eta_0$ and $Y=\eta_0$.  These are normalized Rademacher sums
built from the same signs, and $Y^2=1$.  The claim follows from
\cref{eq:weighted-scalar,eq:unweighted-exponential}.
\end{proof}

\section{Proof of the main theorem}
\label{sec:main-proof}

\begin{proof}[Proof of \cref{thm:main}]
Taking expectation over $D$ in \cref{prop:conditional-unbiasedness} gives
$\E[\widehat{x}]=x$.

The case $y=0$ is immediate, so assume $y\neq0$.
Since $U_i$ is uniform and independent of $D$,
\[
    \E_D[w_i^2\psi_B(z_i)]
    =\E_{D,U_i}\left[w_i^2(\Q_{U_i}(z_i)-z_i)^2\right].
\]
Applying \cref{eq:weighted-scalar-consequence} to the Rademacher pair in
\cref{eq:row-rademacher} therefore gives, uniformly in $i,d,x,y$,
\begin{equation}
    \E_D[w_i^2\psi_B(z_i)]
    \leq
    \left(\frac{3\pi\sqrt{3}}{2}+o_B(1)\right)B^{-2}.
    \label{eq:row-bound}
\end{equation}
Substituting \cref{eq:row-bound} into \cref{eq:rademacher-reduction} and summing
these $d$ bounds proves \cref{eq:main-mse}.

The communication claim follows because each transmitted index belongs to an
alphabet of size $B=2^b$ and hence uses exactly $b$ bits.  The $d$ indices
therefore require exactly $db$ bits; the shared randomness is not transmitted.
\end{proof}

\begin{corollary}[Vector mean-squared error]
\label[corollary]{cor:vector-mse}
Under the hypotheses of \cref{thm:main},
\begin{equation}
    \E\norm{\widehat{x}-x}_2^2
    \leq
    \left(\frac{\pi\sqrt{3}}{2}+o_B(1)\right)B^{-2}.
    \label{eq:vector-mse}
\end{equation}
The convergence is uniform over all $d$ admitting a normalized Hadamard matrix
and all $x\in\Sph^{d-1}$.
\end{corollary}

\begin{proof}
Orthogonality and \cref{eq:z-e-definition} give
\[
    \norm{\widehat{x}-x}_2^2
    =\frac{1}{d}\sum_{i=1}^d e_i^2.
\]
For every fixed row $i$, the first identity in \cref{eq:row-rademacher}
represents $z_i$ as a normalized Rademacher sum under the randomness in $D$.
Moreover, $U_i$ is uniform and independent of $D$, so
\cref{cor:unweighted-scalar} applies to each coordinate.  Averaging proves
\cref{eq:vector-mse}.  Thus the affine dither family matches the
vector-distortion guarantee of Feng et al.~\cite{feng2026provable}.
\end{proof}

\section*{Acknowledgements}

The proof was first obtained using a fully automated Gemini-based agentic
system developed internally at Google. The authors have verified the proof and
edited it for clarity of presentation, and take responsibility for the final
version. The authors would like to thank Ying Feng and Piotr Indyk for helpful
discussions.

\bibliographystyle{alpha}
\bibliography{pairwise_dither_references}

\end{document}